\documentclass[letterpaper, 10 pt, conference]{ieeeconf}  %

\IEEEoverridecommandlockouts                              %
\usepackage{amsmath,amssymb,amsfonts}
\usepackage{bm,bbm,siunitx,hyperref,graphicx}
\usepackage{float}	%
\usepackage{xcolor}
\usepackage{mathtools, nccmath}

\let\labelindent\relax%

\usepackage{enumitem}
\usepackage{makecell}
\usepackage{cellspace}
\usepackage{cite}
\usepackage{balance}   %

\usepackage{amsthm}	%

\newtheorem{theorem}{Theorem}[section]

\DeclareSIUnit \ampereHour {Ah}
\DeclareSIUnit \wattHour {Wh}

\DeclareSIPostPower\powerThreeHalfs{\frac{3}{2}}

\newcommand{\mrb}[1]{\left( #1 \right)} %
\newcommand{\mcb}[1]{\left\{ #1 \right\}} %

\newcommand{\commentOut}[1]{}

\title{\LARGE \bf
Staging energy sources to extend flight time of a multirotor UAV
}

\author{
Karan P. Jain$^1$, Jerry Tang$^1$, Koushil Sreenath$^2$, and Mark W. Mueller$^1$%
\thanks{The authors are with the $^1$High Performance Robotics Lab and $^2$Hybrid Robotics group, Dept. of Mechanical Engineering, UC Berkeley, CA 94720, USA.
{\tt\small \{karanjain, jerrytang, koushils, mwm\}@berkeley.edu}} }%
\begin{document}

\maketitle

\begin{abstract}
Energy sources such as batteries do not decrease in mass after consumption, unlike combustion-based fuels.
We present the concept of staging energy sources, i.e. consuming energy in stages and ejecting used stages, to progressively reduce the mass of aerial vehicles in-flight which reduces power consumption, and consequently increases flight time.
A flight time vs. energy storage mass analysis is presented to show the endurance benefit of staging to multirotors.
We consider two specific problems in discrete staging -- optimal order of staging given a certain number of energy sources, and optimal partitioning of a given energy storage mass budget into a given number of stages.
We then derive results for two continuously staged cases -- an internal combustion engine driving propellers and a rocket engine.
Notably, we show that a multicopter powered by internal combustion has an upper limit on achievable flight time independent of the available fuel mass, but no such limit exists for rocket propulsion.
Lastly, we validate the analysis with flight experiments on a custom two-stage battery-powered quadcopter.
This quadcopter can eject a battery stage after consumption in-flight using a custom-designed mechanism, and continue hovering using the next stage.
The experimental flight times match well with those predicted from the analysis for our vehicle.
We achieve a 19\% increase in flight time using the batteries in two stages as compared to a single stage.
\end{abstract}

\section{Introduction}\label{sec:intro}

The ability to fly as a compact machine has given rise to the use of unmanned aerial vehicles (UAVs) in several applications such as surveillance, mapping, delivery, and search and rescue missions \cite{siebert2014mobile, thiels2015use, erdelj2017help}.
Multirotor UAVs are also being considered for exploration of other worlds \cite{grip2018guidance, lorenz2018dragonfly}.
A fundamental limitation of most UAVs is their flight time -- they must land when their energy source is depleted.
There is a growing demand for higher endurance and range in UAVs with their increasing usage in research and industrial settings.

Innovative approaches have been explored to increase the endurance and range of UAVs.
Broadly, we can classify the approaches into two types: assisted and unassisted.
Assisted methods typically involve the use of fixed stations or mobile vehicles for the replacement or the charging of energy sources.
Battery swapping at a fixed ground station has been presented in \cite{lee2015autonomous, ure2014automated}, and on a mobile ground base has been shown in \cite{barrett2018autonomous}.
Flying replacement batteries to a multirotor using other multirotors is discussed in \cite{jain2019flying}.

Unassisted methods typically involve increasing mechanical or electrical efficiency or the use of optimization-based methods over objectives such as flight time or range.
One such approach is exploiting the efficiency of a fixed-wing and hovering ability of a multirotor by converting them into a hybrid aerial vehicle \cite{saeed2018survey}.
Manipulation of vehicle structure by tilting rotors to increase efficiency is shown in \cite{holda2018tilting}.
An online strategy for optimizing efficiency by altering flight parameters (e.g. speed) over a trajectory is presented in \cite{tagliabue2019model}.
The analysis of a solar-powered UAV, which can potentially fly for a long time, is shown in \cite{reinhardt1996solar}.
However, such a vehicle requires a large wing-span, is sensitive to weather conditions, and is not suitable for indoor settings.

\begin{figure}
	\centering
	\includegraphics[width=0.69\columnwidth]{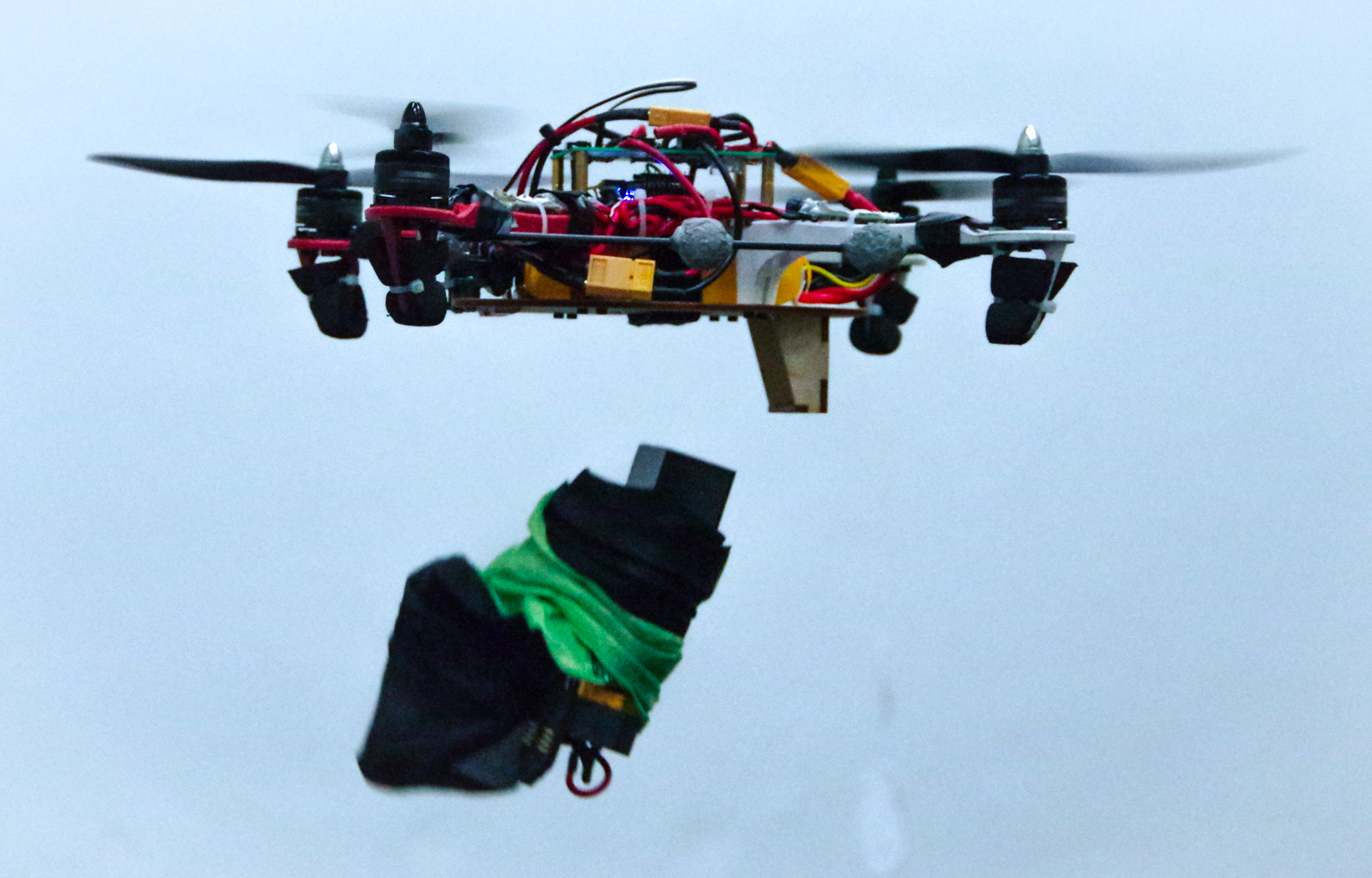}
	\includegraphics[width=0.275\columnwidth]{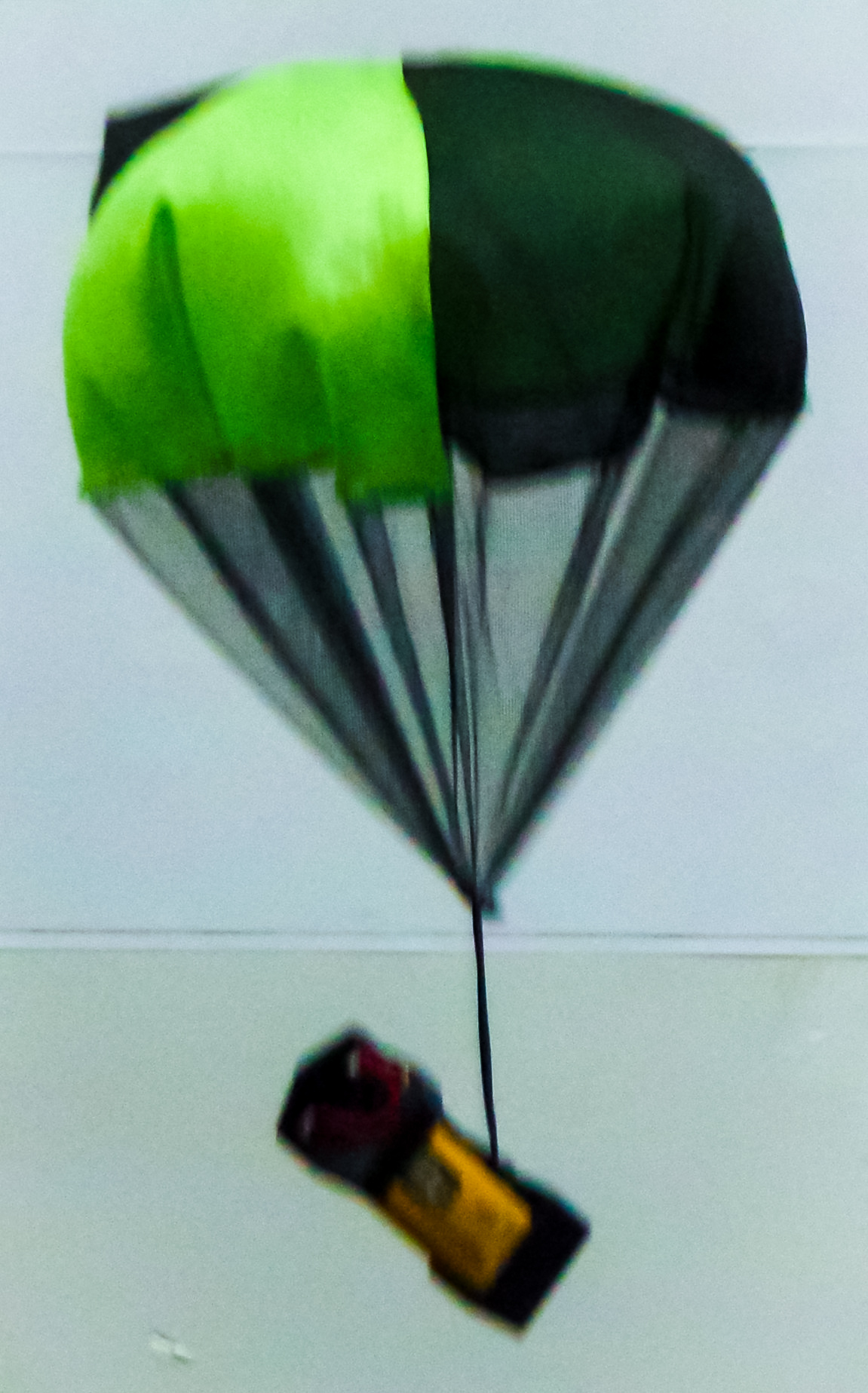}
	\caption{
		(Left:) Quadcopter dropping its first stage after use.
		A parachute is attached to the ejected stage to avoid damage to the battery or the surroundings.
		(Right:) Ejected stage with the deployed parachute.}
	\label{fig:batteryDrop}
\end{figure}

An unassisted approach which is largely unexplored for UAVs is the staging of energy sources.
A possibility of ejecting depleted energy sources is mentioned in \cite{johanningcelldrone}, which presents a modular drone consisting of drone cells, which can potentially be staged.
Energy sources such as batteries are used widely on UAVs, primarily because they are easy to recharge and reuse, and do not produce exhaust, making them a good candidate for indoor settings.
Electric powertrains also have fewer moving parts, making maintenance easy.
Batteries, however, have a clear disadvantage as opposed to combustion fuels, in the sense that the consumed portion still remains as a mass on the vehicle.
This has a significant impact because in the case of electric vehicles, battery mass accounts for a notable fraction of the total mass.
We draw inspiration from multi-stage rockets \cite{hall1958optimization, geckler1960ideal, gray1965cost} which are designed for missions that would otherwise require a much larger single-stage rocket.
We note that rocket staging is about discretely staging the structure, not the fuel.

The contribution of this paper is to analyze, explore, and experimentally validate this energy staging concept for multirotors.
The design space for multirotors with various discrete staging configurations is explored by analyzing the effect on flight time with respect to the energy storage mass.
Flight time performance is also analyzed for continuously staged sources such as an internal combustion engine driving propellers or reaction engines.
Using energy sources in multiple stages as opposed to a single stage is shown to benefit flight time both analytically and experimentally.
Lastly, the environmental impacts and use cases of staging are assessed to realize specialized areas where benefits of staging can be exploited, for example in industrial settings, emergency situations, or for interplanetary exploration missions.

\newcommand{\flightTime}{T_{f}}
\newcommand{\flightTimeArg}[1]{T_{f,#1}}

\newcommand{\flightTimeEqual}{T_\mathrm{flight,eq}}
\newcommand{\flightTimeStage}{T_\mathrm{flight}(k)}
\newcommand{\batteryEnergyStage}{e_k}

\newcommand{\dryMass}{m_d}

\newcommand{\batteryMassTotal}{m_b}
\newcommand{\batteryMassRate}{\dot{m}_b}

\newcommand{\batteryMassTotalRatio}{\phi}
\newcommand{\batteryMassStageArg}[1]{m_{#1}}
\newcommand{\batteryMassStage}{m_k}
\newcommand{\batteryMassRatioStagek}{\theta_k}
\newcommand{\batteryMassRatioStagei}{\theta_i}
\newcommand{\totalMass}{m_\Sigma}
\newcommand{\powerMassRatio}{k_p}

\newcommand{\obj}{J}
\newcommand{\flightTimeA}{T_\mathrm{flight,a}}
\newcommand{\flightTimeB}{T_\mathrm{flight,b}}

\newcommand{\powerElectric}{p_\mathrm{elec}}
\newcommand{\powerAero}{p_\mathrm{aero}}
\newcommand{\power}{p}
\newcommand{\powerConst}{c_p}
\newcommand{\flightTimeConst}{c_T}

\newcommand{\energy}{E}
\newcommand{\specificEnergy}{e_b}
\newcommand{\exhaustVelocity}{v_e}

\section{Staging Analysis}\label{sec:staging}

In this section, we present an analysis of how using an energy source in stages is beneficial for the flight time performance of hovering multirotors.
A brief review is first given of the actuator disk model for power consumption of a propeller; this model is then applied to quantify the gain in flight time by discarding the used energy storage stages in-flight.
This is then compared to two cases where an infinite number of stages exist -- first when the energy source is the fuel combusted (and then exhausted) in an engine that drives rotors, and secondly where the rotors themselves are replaced by rocket motors.

\subsection{Rotor power consumption and flight time}
For an actuator disk model of a propeller that is not translating in the ambient air, the aerodynamic power $\powerAero$ required may be computed \cite{mccormick1995aerodynamics} as
\begin{equation}
  \powerAero = \frac{f^\frac{3}{2}}{\sqrt{2\rho A_p}}
\end{equation}
where $f$ is the thrust produced, $\rho$ is the density of the air, and $A_p$ is the area swept by the rotor. 
This derivation assumes that the flow is inviscid, incompressible, and follows from applying conservation of mass and energy to a control volume containing the propeller. 
Note that translating propellers have more complex relationships, see e.g. \cite{tagliabue2019model}.

The actual power drawn from the energy source to the system will include additional losses, including aerodynamic losses which may be captured in a propeller's figure of merit \cite{mccormick1995aerodynamics} as well as losses in transmission of power to the propeller (e.g. gears or electric resistance). 
We make the simplifying assumption that these losses are all proportional to the power drawn, so that the actual power consumption of a propeller $i$ producing thrust $f_i$ may be captured by
\begin{equation}
\power_i = \powerConst f_i^\frac{3}{2}
\end{equation}
where the constant $\powerConst$ is a function of the propeller size, ambient air density, propeller figure of merit, and powertrain efficiency. 
Note that a similar relationship can be derived from the mechanical power required to drive the propeller shaft under the assumption that the thrust is proportional to the rotational speed of the propeller squared, and that the propeller torque is proportional to the propeller thrust \cite{holda2018tilting}.

For a symmetric quadcopter of mass $m$ to hover, each propeller must produce a force equal to a quarter of the vehicle's weight, and the total power consumption $\power$ is
\begin{equation}
\power = \sum_i \power_i = \frac{1}{2} \powerConst g^\frac{3}{2} m^\frac{3}{2} 
\label{eq:PowerFromMass}
\end{equation}
where the local acceleration due to gravity is given by $g$. 
Let $\specificEnergy$ be the specific energy of the energy source, so the energy budget $\energy$ for a given energy source mass $m_\energy$ is $\energy=\specificEnergy m_\energy$.
Then, at constant total mass $m$, a vehicle with energy source $\energy$ can hover for a flight time of
\begin{equation}
\flightTime = \frac{\energy}{\power} = \specificEnergy \flightTimeConst m_\energy m^{-\frac{3}{2}} \label{eqFlightTimeFromMass}
\end{equation}
where $\flightTimeConst = 2 \powerConst^{-1} g^{-\frac{3}{2}}$ is a vehicle-specific constant, and is specifically independent of the vehicle mass.
In all the subsequent analysis, we assume that all energy storage stages have the same specific-energy.
We note that the energy content need not be strictly proportional to the mass due to additional components such as packaging and connectors.
However, for many commercial batteries, it is shown in \cite{abdilla2015power} that this simplifying assumption is reasonable.

\subsection{Discrete energy storage stages}
A primary disadvantage of batteries for energy storage is that the storage mass does not decrease as the chemical energy is depleted (unlike, for example, a combustion engine). 
From \eqref{eq:PowerFromMass}, it is clear that the power consumption would be reduced if the total mass of the vehicle could be reduced, for example by ejecting parts of the battery as it is depleted.

For a vehicle with $N$ battery stages, where the $i$th battery has mass $\batteryMassStageArg{i}$, the total mass of the vehicle at stage $i$ is given by $\dryMass+\sum_{j=i}^N\batteryMassStageArg{j}$, where $\dryMass$ is the dry mass of the vehicle, i.e. mass of all components of the vehicle excluding energy storage (but including payload, etc.). 
In particular, at stage $i$, the mass $\sum_{j=1}^{i-1}\batteryMassStageArg{j}$ has been ejected.
The flight time for the $i$th stage can be computed from \eqref{eqFlightTimeFromMass} as
\begin{equation}
\flightTimeArg{i} = \specificEnergy \flightTimeConst \batteryMassStageArg{i} \mrb{\dryMass+\sum_{j=i}^N\batteryMassStageArg{j}}^{-\frac{3}{2}}
\end{equation}
with total flight time over all stages
\begin{equation}
\flightTime = \sum_i \flightTimeArg{i}= \specificEnergy \flightTimeConst \sum_{i=1}^{N}\batteryMassStageArg{i} \mrb{\dryMass+\sum_{j=i}^N\batteryMassStageArg{j}}^{-\frac{3}{2}}
\label{eqTotalFlightTimeDiscreteStaged} .
\end{equation}
This equation shows that achievable flight time is directly proportional to the specific energy $\specificEnergy$ and depends on the stage masses $m_i$ in a nonlinear fashion.

\subsubsection{Equal staging}\label{sec:equalStage}
For a vehicle with dry mass $\dryMass$, and total energy storage mass $\batteryMassTotal$ split equally over $N$ stages, the flight time can be computed from \eqref{eqTotalFlightTimeDiscreteStaged} as
\begin{equation}
\flightTime = \frac{\specificEnergy \, \flightTimeConst \, \batteryMassTotal}{N} \sum_{i=1}^N \mrb{\dryMass + \frac{i}{N}\batteryMassTotal}^{-\frac{3}{2}} .
\label{eqFlightTimeFromMassEquallyStaged}
\end{equation}

Fig.~\ref{fig:equalStage} shows plots of (normalized) flight time vs. the ratio of total energy storage mass ($\batteryMassTotal$) to total initial vehicle mass ($\dryMass+\batteryMassTotal$), for various number of stages.
The case of equal staging is plotted using solid curves.
The normalization factor is the maximum achievable flight time for an infinitely-staged vehicle as derived in Section~\ref{sec:contStaging}.
Since the figure is normalized, it is valid for a multirotor with a different powertrain efficiency, or one flying on another planet with a different surface gravity and air density.

The plots in Fig.~\ref{fig:equalStage} and the results from equal staging can be used in the following ways.
For a given energy storage mass, one can fix the x-axis value, and decide the number of stages they want based on their flight time requirements.
On the other hand, the choice for designers might be flight time for their vehicle.
In that case, they can fix the y-axis value, and then choose the number of stages based on the mass constraints.

\subsubsection{Optimal staging order}\label{sec:optimalStagingOrder}
We consider the case of having a series of $N$ energy storage stages, of known, fixed but different masses, with the only design variable being the order in which to stage them.

\begin{theorem}
\label{thm:optimalStagingOrder}
The optimal flight time is achieved by staging in order of decreasing mass, so that the heaviest stage is depleted and discarded first.
\end{theorem}

\begin{proof}
The proof follows by contradiction.
Assume that an optimal staging sequence is given by ${\mathcal{S}}^*=\mcb{m_1, m_2, \ldots, m_N}$ where for some value $k\in\mcb{1,\ldots,N}$ the stage $k$ is lighter than the following stage $k+1$, i.e. $m_k<m_{k+1}$.
Let ${T}^*$ represent the total flight time, computed with \eqref{eqTotalFlightTimeDiscreteStaged} for this sequence.

Let $\bar{\mathcal{S}} := \mcb{m_1, \ldots, m_{k-1}, m_{k+1}, m_{k}, m_{k+2}, \ldots, m_N}$ represent a modified staging order, which only interchanges the order of the original $k$th and $\left(k+1\right)$th stage.
Let $\bar{T}$ represent the associated total flight time.
We note that flight times from the first $k-1$ and the stages continuing after $k+2$ are identical.
We denote the short-hand $M:= \dryMass + \sum_{i=k+2}^N m_i$ as the total mass of the vehicle after discarding the first stages $(k+1)$ stages.
Then the difference in flight times is 
\begin{align}
\frac{{T}^*-\bar{T}}{\specificEnergy \flightTimeConst} = &  \frac{m_k}{\mrb{M + m_k + m_{k+1}}^{\frac{3}{2}}} + \frac{m_{k+1}}{\mrb{M + m_{k+1}}^{\frac{3}{2}}}\nonumber
\\ & - \frac{m_{k+1}}{\mrb{M + m_k + m_{k+1}}^{\frac{3}{2}}} - \frac{m_{k}}{\mrb{M + m_{k}}^{\frac{3}{2}}}
\\=&  \mrb{m_{k}-m_{k+1}} q\mrb{m_{k} + m_{k+1}}  \nonumber
\\&  + m_{k+1} q\mrb{m_{k+1}} - m_{k} q\mrb{m_{k}}
\label{eq:deltaTflightOne}
\end{align}
where we define $q\mrb{m}:=\mrb{M + m}^{-\frac{3}{2}}$ for convenience.
We note that
\begin{align}
  \tfrac{\partial}{\partial m} q\mrb{m} &=: q'\mrb{m} = -\tfrac{3}{2}\mrb{M+m}^{-\frac{5}{2}} < 0 ,
\\\tfrac{\partial^2 }{\partial m^2} q\mrb{m} &= \tfrac{15}{4}\mrb{M+m}^{-\frac{7}{2}} > 0 .
\end{align}

We use Lagrange's mean value theorem to relate the average slope of $q\mrb{m}$ between points $\{m_k, m_{k+1}\}$ with the average slope between points $\{m_{k+1}, m_k+m_{k+1}\}$ as
\begin{align}
\begin{split}
\frac{q\mrb{m_{k+1}}-q\mrb{m_k}}{m_{k+1}-m_k} &< q'\mrb{m_{k+1}} 
\\ & < \frac{q\mrb{m_k+m_{k+1}}-q\mrb{m_{k+1}}}{\mrb{m_k+m_{k+1}} - m_{k+1}}.
\end{split}
\label{eq:LMVT}
\end{align}
Rearranging this, and using $m_{k+1}>m_k>0$ gives
\begin{align}
& \mrb{m_k-m_{k+1}} q\mrb{m_k + m_{k+1}} \nonumber
\\ & + m_{k+1} q\mrb{m_{k+1}} - m_k q\mrb{m_k}  < 0 .
\end{align}
Substituting this in \eqref{eq:deltaTflightOne} gives $\bar{T}>T^*$.
However, this contradicts the assumption that $T^*$ is the optimal flight time.
Thus, $S^*$ cannot be an optimal staging sequence, and the optimal staging sequence has $m_1 \geq m_2 \geq \ldots \geq m_N$, proceeding from the heaviest stage first to lightest stage last.
\end{proof}

\subsubsection{Optimal mass partitioning}\label{sec:optPartition}
Given a dry mass $\dryMass$, an energy storage mass budget $\batteryMassTotal$, and a total number of stages $N$ to be used, it is of interest to find the optimal stage masses that add up to $\batteryMassTotal$, to maximize the flight time.

This is equivalent to maximizing $\flightTime$ in \eqref{eqTotalFlightTimeDiscreteStaged} over the decision variables $m_i$, under the inequality constraints $m_i~>~0$ and the equality constraint $\sum_i m_i = \batteryMassTotal$.

We define new decision variables, $x_i = \dryMass+\sum_{j=i}^N\batteryMassStageArg{j}$, so that $m_i = x_{i}-x_{i+1}$.
Substituting this in \eqref{eqTotalFlightTimeDiscreteStaged}, the flight time can be written in terms of the decision variables as
\begin{equation}
J = \sum_{i=1}^{N} \frac{x_{i} - x_{i+1}}{x_{i}^{\frac{3}{2}}}
\end{equation}
where $J$ is now our objective function to be maximized.
The constants $\specificEnergy$ and $\flightTimeConst$ are dropped because they are scaling factors which do not affect the decision variables.

This can be formulated as a constrained optimization problem as follows,
\begin{align}
\begin{split}
\max_{x_1, x_2, \cdots, x_{N+1}} \quad & J \\
\textrm{s.t.} \quad & x_1 = \dryMass + \batteryMassTotal , \\
& x_{N+1} = \dryMass , \\
& x_{i+1} < x_{i} \quad \textrm{for} \quad i = 1, 2, \cdots, N .
\end{split}
\label{eq:optimalPartitioningProblem}
\end{align}

Solving the above problem using the KKT conditions \cite{boyd2004convex}, we obtain a solution in the form of simultaneous nonlinear equations shown below,
\begin{align}
\begin{split}
\frac{1}{x_{i}^\frac{3}{2}} + \frac{2}{x_{i-1}^\frac{3}{2}} - \frac{3 x_{i+1}}{x_{i}^\frac{5}{2}} &= 0 \quad \textrm{for} \quad i=2,3,\cdots,N , \\
x_1 &= \dryMass + \batteryMassTotal , \\
x_{N+1} &= \dryMass .
\end{split}
\label{eq:optimalPartitioningSolution}
\end{align}

\begin{figure}
	\centering
	\includegraphics[width=\columnwidth]{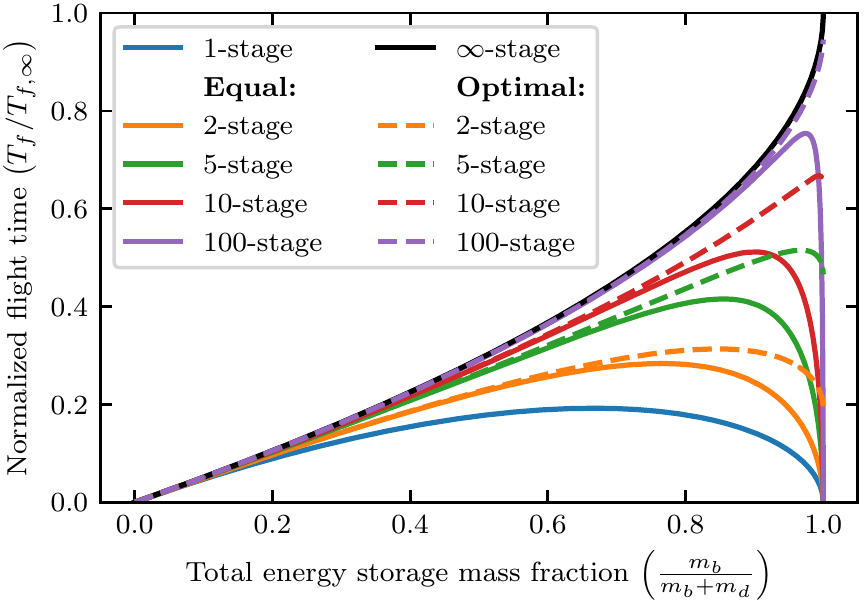}
	\caption{
		Effect of total energy storage mass on hovering flight time for various number of stages.
		Dry mass and total energy storage mass are respectively denoted by $\dryMass$ and $\batteryMassTotal$.
		Note that this plot is valid only for $\dryMass>0$, and the energy storage fraction reaches $1.0$ in the limit $\batteryMassTotal \rightarrow \infty$.
		Solid lines are for equally staged energy sources (see Section~\ref{sec:equalStage}), whereas dashed lines are for optimally staged (see Section~\ref{sec:optPartition}).
		The flight times are normalized with respect to that for continuously staged case with infinite energy storage mass $\flightTimeArg{\infty}$ (see \eqref{eq:flightTimeInfty}).
		Normalization makes this plot valid for any multirotor flying in any environment.
	}
	\label{fig:equalStage}
\end{figure}

Even for the simplest case of $N=2$, this requires the roots of a fifth-degree polynomial, for which no closed-form expression exists.
Nonetheless, the simultaneous equations can be solved numerically.
Plots for a few sample $N$ values are shown in Fig.~\ref{fig:equalStage} as dashed lines to compare with the equally staged case (solid lines).
For $N=$ 2, 3, 4, and 5, the predicted maximum achievable flight time is higher by 10.5\%, 16.9\%, 21.1\%, and 24.0\% respectively for optimal staging as compared to equal staging.

From the discrete staging analysis, we conclude that staging can be beneficial in the following ways.
First, for the same amount of energy, we can increase the flight time as compared to a single-stage vehicle.
This is to be expected because as the multi-stage vehicle ejects stages, its mass and consequently power consumption is reduced.
Since the total energy is the same in both cases, the flight time of the multi-stage vehicle is higher.
Second, to achieve the same flight time, a multi-stage vehicle will be lighter and more compact, by a similar argument, which could make it safer.
Third, in the case of any multi-stage UAV, the mass and moment of inertia reduces with every stage ejected, increasing its agility.

\subsection{Continuous staging}\label{sec:contStaging}
A vehicle powered by combustion, where the combustion products are exhausted, can be described as a limiting case of the above, with continuous staging. 
We consider two forms of this: using a combustion engine to provide power to rotors, and a reaction (e.g. rocket) engine.

\subsubsection{Internal combustion engines}
When burning (and exhausting) fuel to perform the mechanical work of driving propellers, the remaining mass of the energy storage $m_\energy$ will evolve at a rate proportional to the power consumption,
\begin{equation}
\dot{m}_\energy = -\frac{1}{\specificEnergy} \dot{\energy} = -  \frac{\powerConst g^\frac{3}{2}}{2\specificEnergy}  \mrb{\dryMass+m_\energy}^\frac{3}{2}
\end{equation}
where we've used \eqref{eq:PowerFromMass} and again assumed a symmetric quadcopter at hover.
Solving this, and substituting $m_\energy\mrb{0}=\batteryMassTotal$, and $m_\energy\mrb{\flightTime}=0$ gives the total flight time $\flightTime$ as
\begin{equation}
\flightTime = \frac{4\specificEnergy}{\powerConst g^\frac{3}{2} \sqrt{\dryMass}}\mrb{1-\mrb{1+\frac{\batteryMassTotal}{\dryMass}}^{-\frac{1}{2}}} .
\end{equation}
This flight time value can also be derived by taking the limit as $N \rightarrow \infty$ in \eqref{eqFlightTimeFromMassEquallyStaged}.

Note that the flight time is, as may be expected, monotonically increasing in the initial fuel mass $\batteryMassTotal$. 
However, there exists a natural upper limit to achievable flight time, even for arbitrarily large quantities of fuel:
\begin{equation}
\flightTime < \flightTimeArg{\infty} = \lim_{\batteryMassTotal \rightarrow \infty} \flightTime = \frac{4\specificEnergy}{\powerConst g^\frac{3}{2} \sqrt{\dryMass}} = \frac{2 \specificEnergy \flightTimeConst}{\sqrt{\dryMass}}.
\label{eq:flightTimeInfty}
\end{equation}
This reveals a fundamental limit of rotor propulsion. Even under the best case of continuous staging, achievable flight time has an upper limit no matter how much fuel is used, for a specific vehicle (fixed $\powerConst$), energy source (fixed $\specificEnergy$), fixed environment and fixed dry mass.
In practice, other constraints would act to limit achievable performance, e.g. thrust and structural limitations.

We also note that as the dry mass is reduced, the flight time increases, with the limit $\flightTime \rightarrow \infty$ as $\dryMass \rightarrow 0$.

\subsubsection{Reaction engines}\label{sec:reactionEngine}
Another method of achieving continuous staging is reaction engines (such as rocket motors), where propellant is expelled at high velocity, producing thrust in response to the momentum flux leaving the engine. 
A simple analysis, assuming that fuel is expelled at a constant exhaust velocity $\exhaustVelocity$, gives a thrust force as a function of the fuel mass flow rate $\dot{m}_f$ as
\begin{equation}
f = \dot{m}_f \exhaustVelocity .
\end{equation}
For a quadcopter-like rocket ship, with four rocket engines symmetric around the center of mass, we have that 
\begin{equation}
4 \dot{m}_f \exhaustVelocity = -g\mrb{\dryMass + m_f} .
\end{equation}
Solving this, and setting again $m_f(0)=\batteryMassTotal$, the achievable hover flight time is 
\begin{equation}
\flightTimeArg{\mathrm{rocket}} = \frac{4 \exhaustVelocity}{g} \ln \mrb{1+\frac{\batteryMassTotal}{\dryMass}} .
\label{eq:flightTimeReactionEngine}
\end{equation}
This result is very closely linked to the Tsiolkovsky rocket equation \cite{turner2008rocket} which describes the achievable velocity change as a function of fuel burnt, with the same functional form.

We note that this achievable flight time will grow unbounded with increasing fuel mass $\batteryMassTotal$.
Thus, rocket propulsion does not suffer from the same fundamental limitation as aerodynamic propulsion. 
This suggests that, for certain extreme design specifications, it may be preferable to create a rocket-propelled multirotor-style robot, rather than using aerodynamic propulsion.

Comparing the propeller thrust relation of \eqref{eqFlightTimeFromMass} to that for reaction engines \eqref{eq:flightTimeReactionEngine}, it is notable that the exponent for gravity is different.
Combining this with the dependence on air density, the analysis suggests that for environments with high density and low gravity (e.g. Titan), aerodynamic propulsion might be the preferred choice \cite{lorenz2018dragonfly}.
But, for environments with low air density (e.g. Mars), a rocket-propelled system may have superior flight time for a given system mass and size.

\section{Experimental Hardware Design}\label{sec:design}

In this section, we explain the design of the quadcopter used in our experiments, the battery staging mechanism, and the battery switching circuit.

\subsection{Vehicle design}\label{sec:vehDesign}

The quadcopter is designed to have enough payload capacity for carrying useful sensors such as surveillance cameras, or environmental sensors.
Its dry mass is \SI{565}{\gram}, and it can generate a maximum thrust of \SI{27}{\newton}.
Its arm length is \SI{165}{\milli \meter}, and it uses four \SI{203}{\milli \meter} diameter propellers.
We stack a battery switching circuit on top of the quadcopter for staging.
The quadcopter is powered by two batteries. The first-stage battery is placed at the bottom of the quadcopter that can be ejected when depleted, and the second-stage battery is placed at the center of the quadcopter which always stays onboard. Copper plates mounted at the bottom serve as input leads from the first-stage battery.
Fig.~\ref{fig:batteryDrop} shows a picture of the quadcopter.
We use two types of batteries -- a ``heavy'' 3S {\SI{2.2}{\ampereHour}} lithium polymer (LiPo) battery weighing {\SI{190}{\gram}}, and a ``light'' 3S {\SI{1.5}{\ampereHour}} LiPo battery weighing {\SI{135}{\gram}}.
We use the heavy battery in both the first and second stages for experimental comparison between staged and unstaged cases.
To verify the optimal order of staging, we use the heavy and light battery in the two possible permutations.

\subsection{Staging mechanism}\label{sec:stageMech}

\begin{figure}
	\centering
	\includegraphics[width=0.9\columnwidth]{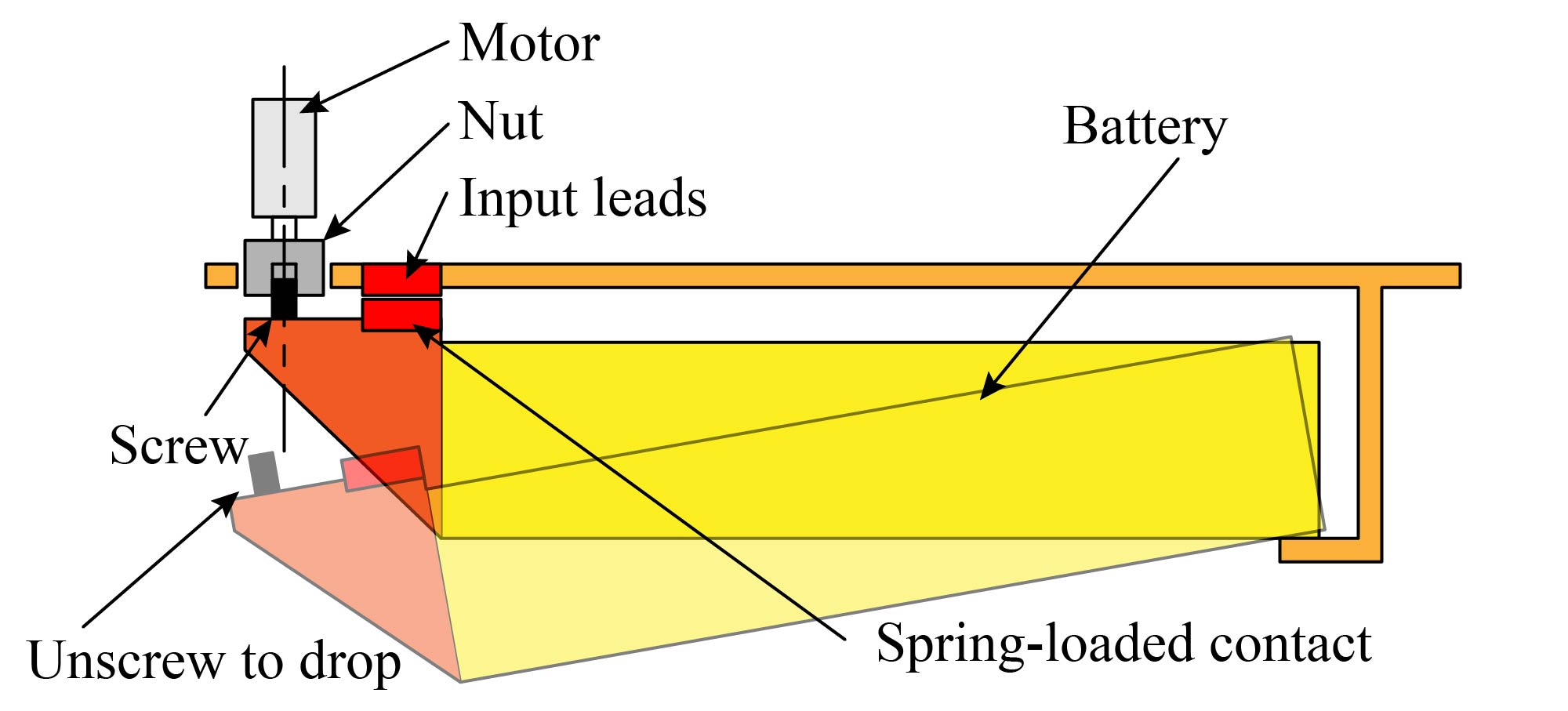}\\
	\vspace{5mm}
	\includegraphics[width=0.9\columnwidth]{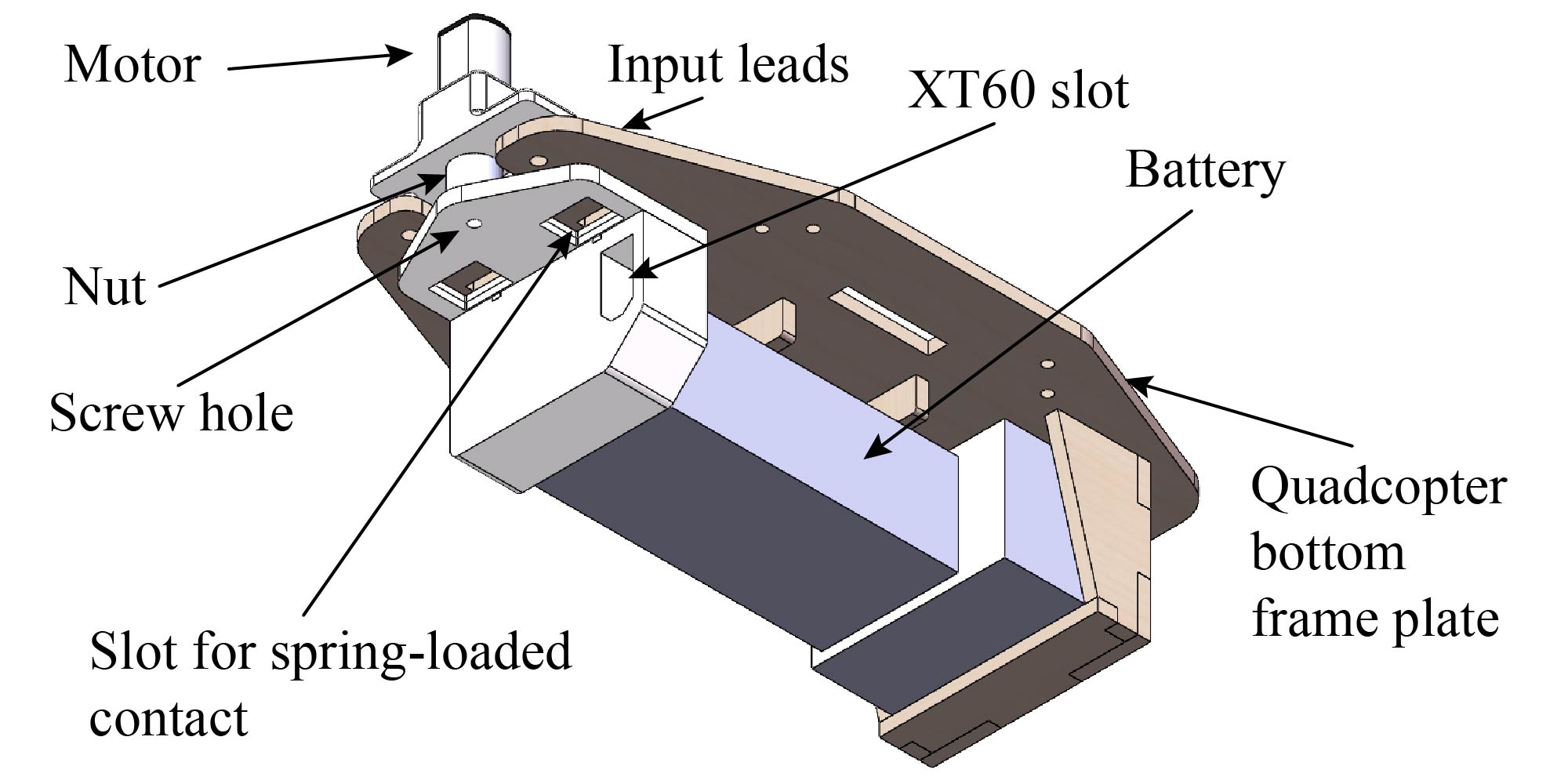}
	\caption{
		(Top:) Schematic of the battery-dropping mechanism.
		(Bottom:) CAD model of the mechanism.
	}
	\label{fig:battDropMech}
\end{figure}

The staging mechanism is shown in Fig.~\ref{fig:battDropMech}.
The mechanism is placed on the lower section of the quadcopter frame.
In order for the first-stage battery to engage with the mechanism, we have an attachment on one end of the battery.
This attachment has a bolt, and two spring-loaded contacts, which are electrically connected to the output leads of the first-stage battery.
On the side of the quadcopter frame, a ``detaching'' motor is installed, which has a nut fixed on its shaft.
Next to the motor are two copper plates that serve as input leads from the first-stage battery.
The staging mechanism adds a total of \SI{45}{\gram} mass to the quadcopter -- \SI{30}{\gram} to the first-stage battery and \SI{15}{\gram} to the quadcopter frame.
This corresponds to a mass increase of 8\%.

To connect the first-stage battery, we engage the bolt from the battery with the nut on the motor shaft.
This connects the spring-loaded contacts to the input leads.
To disconnect the battery, we command the motor to unscrew the bolt.
When the bolt is completely disengaged from the nut, the battery loses support and thus drops, and then softly lands using an attached parachute.

\subsection{Battery switching circuit}\label{sec:batterySwitch}

The design of our battery switching circuit is inspired by \cite{jain2019flying}.
Since our system is flying, we cannot afford to cut the power supply when switching from the first battery to the second battery.
The two batteries need to be connected in parallel for some time to achieve this.
We connect diodes in series with each of the batteries to avoid reverse currents due to the voltage difference between the batteries.

A normally closed relay is connected in series with the first battery.
By opening the switch, we can draw power from the first battery even when it is at a lower voltage than the second battery.
This is necessary to completely consume the first battery, before starting the use of the second battery.
The relay coil is connected across the first battery input leads in series with a MOSFET.
The gate terminal of the MOSFET is connected to a GPIO pin on the flight controller to control the relay.
Fig.~\ref{fig:battSwitcher} shows a schematic diagram of the battery switching circuit.
The circuit weighs \SI{60}{\gram}.
We note that this circuit is built on a prototype board and can potentially be made much lighter using smaller electronics on a PCB.

\begin{figure}
	\centering
	\includegraphics[width=\columnwidth]{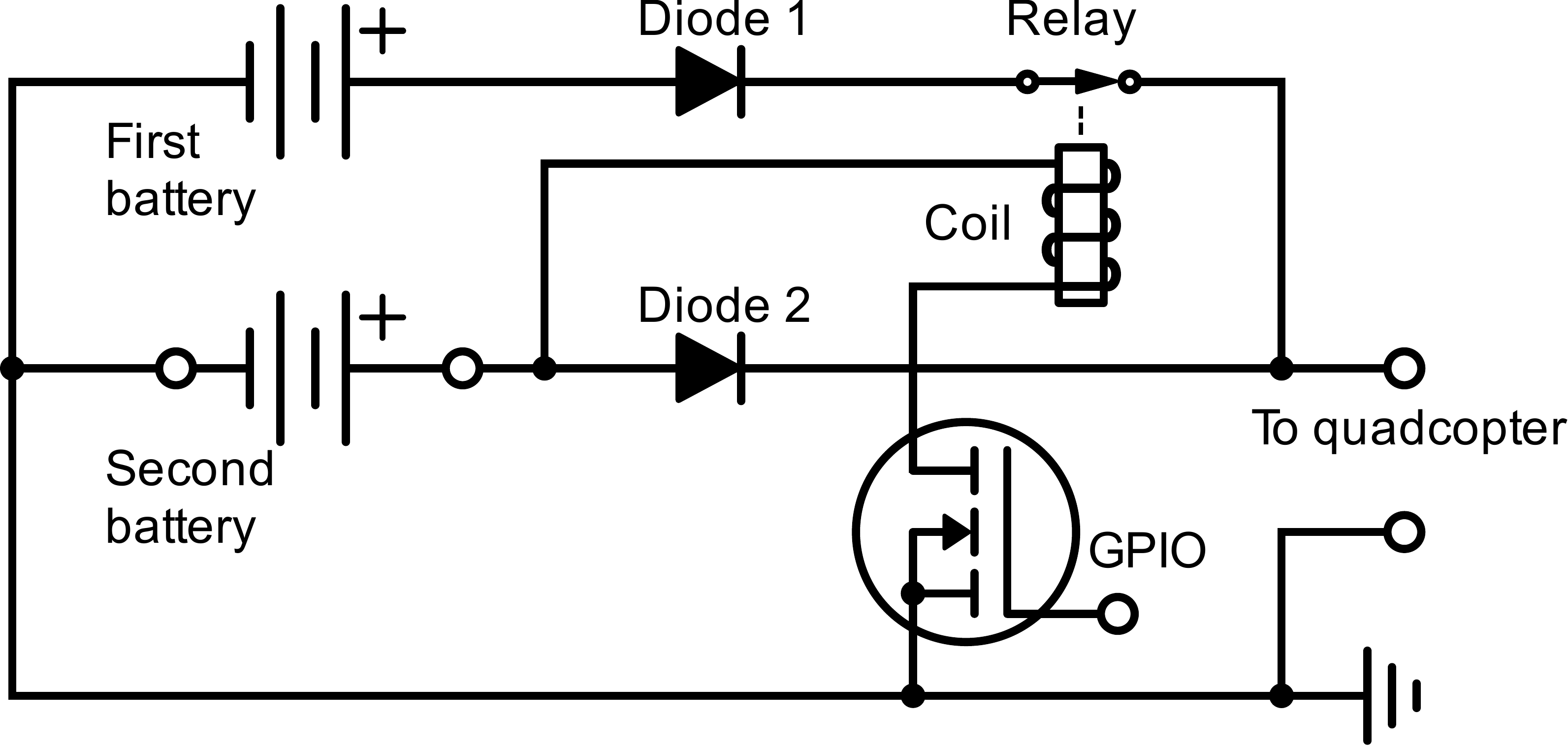}
	\caption{Schematic of the battery switching circuit.}
	\label{fig:battSwitcher}
\end{figure}

\section{Experimental Validation}\label{sec:experiment}

We validate the staging analysis by first predicting the flight times for our quadcopter with and without staging.
We then conduct hovering flight experiments for single-stage and two-stage cases.
The flight time for the two-stage case is shown to be higher.
The optimal order of staging is also validated using the heavy and light batteries in two different orders.
All experimental flight times are also compared with the predictions from the analysis and are shown to match well, validating the analysis.
An experiment with the quadcopter ejecting a stage while traversing a circular trajectory is also conducted to show that the effect on maneuvering is negligible.
This can be seen in the video attachment.

\subsection{Flight time prediction from analysis}

We use the results from equal staging (Section~\ref{sec:equalStage}) since we use the same type of battery (heavy) for both stages.
The values of specific energy and $\flightTimeConst$ are determined empirically from flight experiments on the quadcopter, which respectively are $\specificEnergy=\SI{130}{\wattHour \per \kilogram}$ and $\flightTimeConst = \SI{6.2e-3}{\kilogram \powerThreeHalfs \per \watt}$.
Using Eq.~\eqref{eqFlightTimeFromMassEquallyStaged} with $\batteryMassTotal=\SI{380}{\gram}$ and $\dryMass={\SI{595}{\gram}}$ (includes additional {\SI{30}{\gram}} from the first-stage battery), we predict flight times of {\SI{19.1}{\minute}} for the single-stage case, and {\SI{22.8}{\minute}} for the two-stage case.

\subsection{Experimental setup}

The quadcopter used in our experiments is localized via sensor fusion of a motion capture system and an onboard rate gyroscope \cite{jain2019modeling}.
Experimental data from the motion capture system, voltage sensor, and the current sensor is logged via radio for post-processing.
We control the quadcopter using a cascaded position and attitude controller shown in Fig.~\ref{fig:ctrler}.
The measured battery voltage is used to decide on commands for battery switching, stage ejection, and landing.

\begin{figure}
	\centering
	\includegraphics[width=\columnwidth]{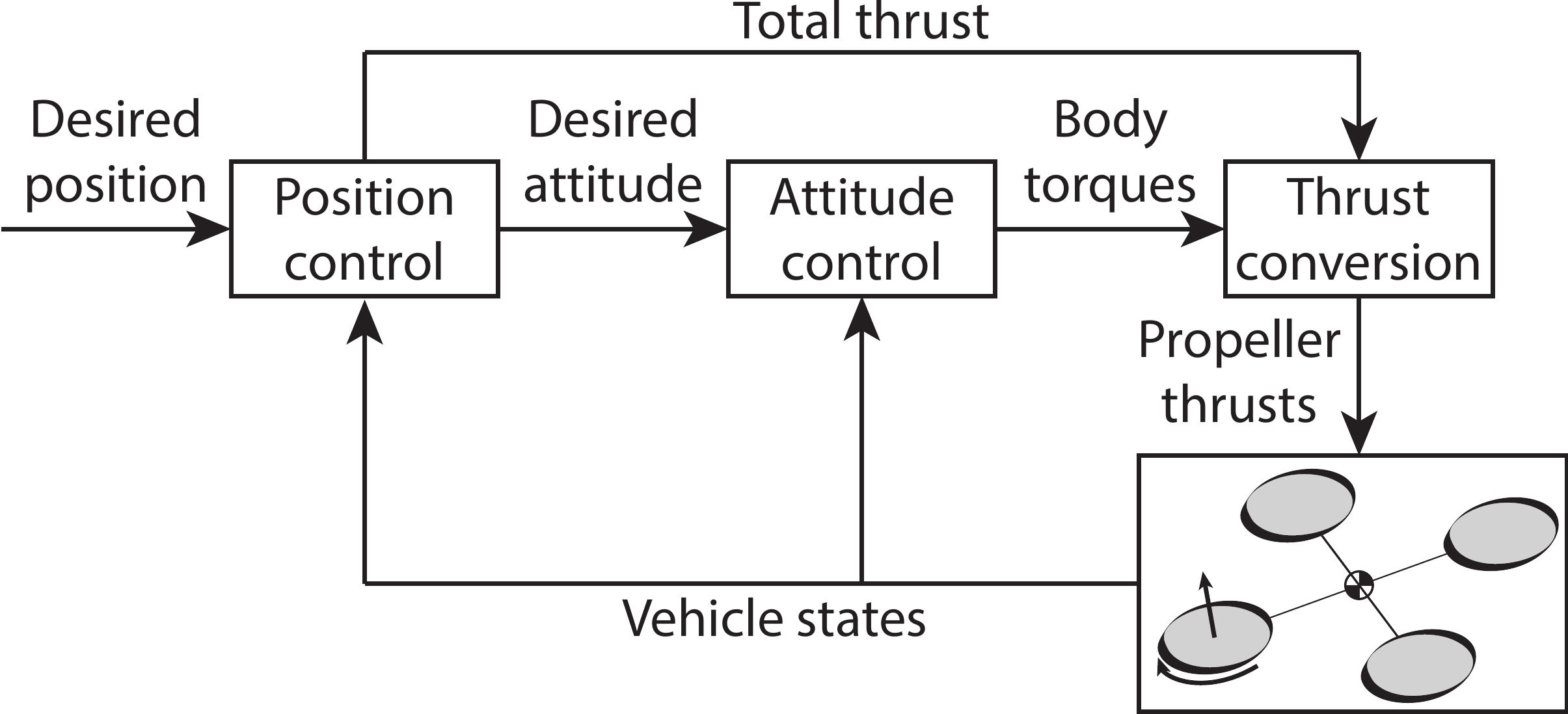}
	\caption{Block diagram of the quadcopter controller.}
	\label{fig:ctrler}
\end{figure}

\subsection{Demonstration}\label{sec:demonstration}

\begin{figure*}
	\centering
	\includegraphics[width=\textwidth]{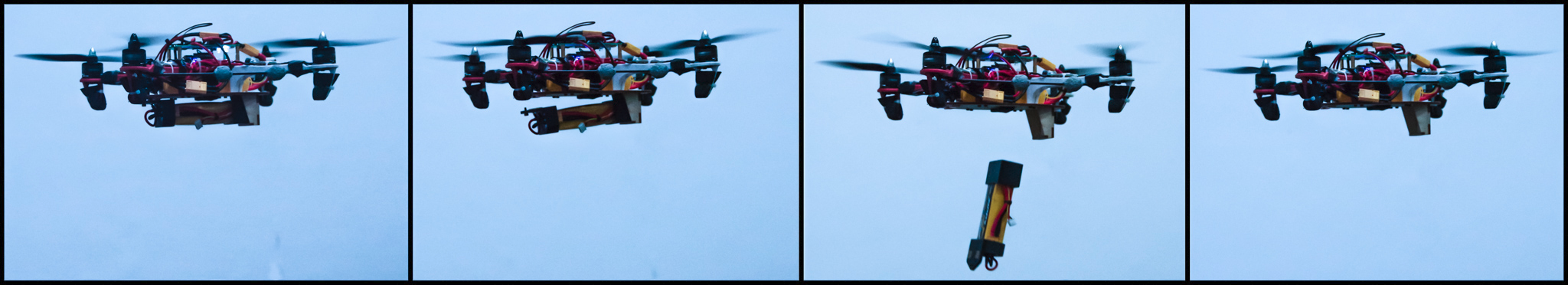}
	\caption{Sequence of images from the two-stage flight experiment. From left to right:
		(a) Quadcopter hovers using only first-stage battery.
		(b) First-stage battery is depleted. Quadcopter starts using second-stage battery. The detaching motor is simultaneously activated to disengage the bolt. First stage is about to fall.
		(c) First stage is ejected and falls towards the ground.
		(d) Quadcopter continues hovering using the second-stage battery, with reduced mass and power consumption.
		\textbf{Note:} We have not attached a parachute in this case to clearly show how the staging works.
		The first stage fall is broken softly by a net below the quadcopter.
		In real-life settings, it is advised to use some fall breaking mechanism to not damage the energy source or any property (for e.g. by using a parachute as shown in Fig.~\ref{fig:batteryDrop}).
		All reported experimental flight times and plots use a parachute in the two-stage case.
	}
	\label{fig:batteryDropFilmReel}
\end{figure*}

The benefit of staging batteries is demonstrated by conducting two types of experiments (shown in the video attachment): one without staging (single-stage), and the other with staging (two stages).
We conduct three flight experiments for each type.
All of our experiments consume the batteries from a fully charged state of {\SI{4.2}{\volt}} per cell to a fully discharged state of {\SI{3.0}{\volt}} per cell.
For the 3 cell batteries that we use, the corresponding voltages are {\SI{12.6}{\volt}} and {\SI{9.0}{\volt}}.

In the single-stage experiments, the quadcopter hovers using the two batteries simultaneously until both the batteries are completely discharged.
In the three single-stage experiments, the quadcopter hovered for an average time of {\SI{19.16}{\minute}}, with a standard deviation of {\SI{0.18}{\minute}}.

In the two-stage experiments, the quadcopter initially hovers using the first battery.
Once the first battery is discharged, the second battery is connected to the quadcopter by closing the relay in the switching circuit.
The first stage is then detached by activating the detaching motor in the staging mechanism.
Hovering continues until we completely consume the second battery.
Fig.~\ref{fig:batteryDropFilmReel} shows a sequence of snapshots from the two-stage experiment.
In the three two-stage experiments, the quadcopter hovered for an average time of {\SI{22.80}{\minute}}, with a standard deviation of {\SI{0.17}{\minute}}.

\subsection{Validation of optimal order}\label{sec:optOrderValidation}
We experimentally validate the optimal order of staging result from Section~\ref{sec:optimalStagingOrder} by conducting flight experiments on the same quadcopter using the heavy and light batteries.
These batteries had an experimentally determined specific energy of $\specificEnergy={\SI{120}{\wattHour \per \kilogram}}$.

Using \eqref{eqTotalFlightTimeDiscreteStaged} with same $\flightTimeConst$ as before, and $\dryMass={\SI{595}{\gram}}$, we predict flight times of {\SI{19.0}{\minute}} with $m_1={\SI{135}{\gram}}$ and $m_2={\SI{190}{\gram}}$ (light battery on first-stage), and {\SI{19.3}{\minute}} with $m_1={\SI{190}{\gram}}$ and $m_2={\SI{135}{\gram}}$ (heavy battery on first-stage).
As expected from the analysis of optimal order, we get a higher flight time when the heavier battery is ejected first.

Three flight experiments were conducted for each order.
We observed an average flight time of {\SI{19.00}{\minute}} when the light battery was ejected first, with a standard deviation of {\SI{0.14}{\minute}}.
When the heavy battery was ejected first, the average flight time was {\SI{19.47}{\minute}}, with a standard deviation of {\SI{0.32}{\minute}}.
These experiments are in agreement with the predictions, validating the optimal staging order analysis.

\subsection{Discussion}\label{sec:discussion}
The plots of input voltage and power vs. time for sample experiments from the single-stage and two-stage case are shown in Fig.~\ref{fig:voltagePower}.
We see the characteristic LiPo battery discharge curve \cite{navarathinam2011characterization} in each of the voltage vs. time plots.

\begin{figure}
	\centering
	\includegraphics[width=\columnwidth]{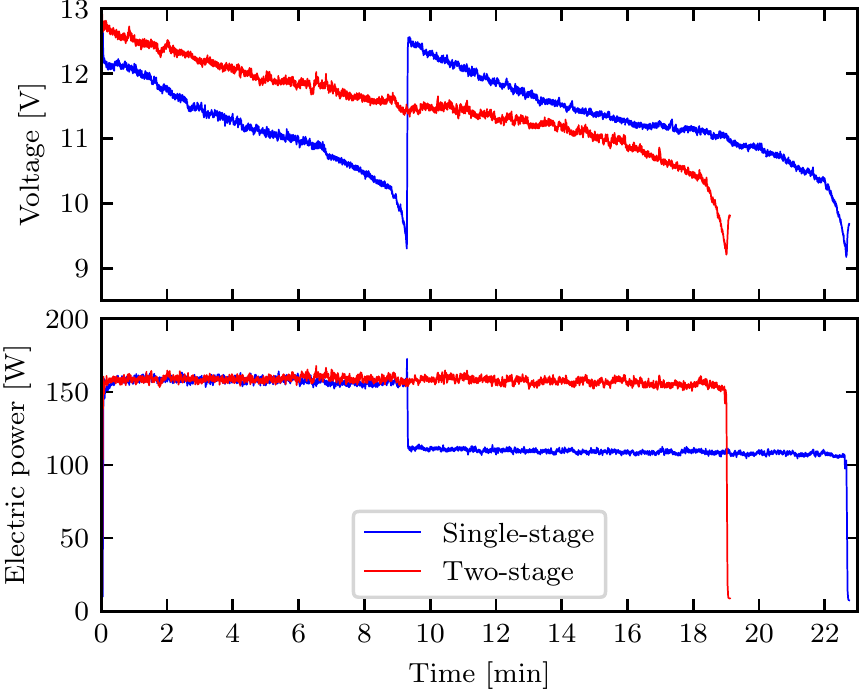}
	\caption{
		Input voltage and power vs. time measured at the electrical input to the quadcopter in typical experiments with the heavy batteries.
	}
	\label{fig:voltagePower}
\end{figure}

We observe that power input to the quadcopter remains approximately constant, as long as its mass is not changing.
For the single-stage experiments, the mass is constant for the entire experiment.
For the two-stage experiments, the power input remains approximately constant until the first battery is completely discharged.
Once we detach the first stage, the quadcopter mass is lower and the power consumption settles to a lower value corresponding to the reduced mass.

An important observation here is that energy consumption (calculated from power vs. time data) is similar for the single-stage and two-stage experiments with the heavy batteries.
This is expected since we have the same total energy in all cases.
The average energy consumption in the single-stage experiments is \SI{1.819e5}{\joule} and in the two-stage experiments is \SI{1.761e5}{\joule}.
We note that the values mentioned here are the electrical energy utilized by the quadcopter.
In the single-stage experiments, we are using two batteries in parallel which, in essence, halves the internal resistance, as compared to the two-stage experiments, where we are using one battery at a time.
This increases the electrical efficiency, and thus the quadcopter is able to utilize more output energy for the same input energy in the single-stage experiments.

The flight time predictions match the experimental values within 5\%.
This validates that the assumptions used in the staging analysis are, in fact, realistic.

To be fair in comparing the flight times, we scale up the flight time of the single-stage experiments by a factor to account for the additional mass of the staging mechanism.
This factor is obtained from the staging analysis by taking a ratio of the flight times with $\dryMass={\SI{550}{\gram}}$ to flight time with $\dryMass={\SI{595}{\gram}}$.
Using $\batteryMassTotal={\SI{380}{\gram}}$, the factor equals $1.07$.
The scaled single-stage flight time is {\SI{20.56}{\minute}}.
This shows that using the battery in stages is beneficial for flight time, even with the additional mass of the staging mechanism.

For an illustration, we compute the battery mass required for a single-stage vehicle to give us the same flight time as is obtained with two stages.
Using Eq.~{\eqref{eqFlightTimeFromMassEquallyStaged}} with $\flightTime={\SI{22.8}{\minute}}$, $\dryMass={\SI{550}{\gram}}$ (we exclude staging mechanism's mass), $N=1$, and $\specificEnergy$ and $\flightTimeConst$ same as before, we calculate $\batteryMassTotal={\SI{525}{\gram}}$.
So a single-stage vehicle would need an additional {\SI{145}{\gram}} of battery to achieve the same flight time as a two-stage vehicle.
Even considering the additional mass added by the staging mechanism, we save about {\SI{100}{\gram}} or about 18\% of the dry mass by using just one additional stage.

As a consequence of the increased flight time, the advantage of our proposed staging approach is a single flight, and hence a single continuous stream of data from sensors.
Furthermore, any additional flight time could potentially lead to using just one multirotor instead of two or more with identical sensors, which is cost effective.

We note that the potential gains of using more stages are even higher, with certain flight times not achievable by a single-stage vehicle, as seen in Fig.~\ref{fig:equalStage}.
Since higher flight times can be achieved via staging without adding as much battery mass as required for a single stage, usually none or minimal changes to the vehicle's structure and powertrain would be required.
Moreover, after each battery ejection, the inertia of the vehicle reduces, thereby increasing its agility.

\vspace{0.1mm}
\section{Environmental Impact and Use Cases}\label{sec:envImpact}
It is important to note that there are potentially severe environmental effects to take into consideration when utilizing staged batteries, which limit their utility to special cases.
Batteries ejected in mid-flight pose an immediate safety risk due to their momentum and kinetic energy when they reach the ground.
Use of parachutes, as in the proof-of-concept in this paper, reduces this risk but does not eliminate it, and moreover may make it harder to predict where the battery reaches the ground as external disturbances from the wind have a larger effect.
Another major concern is the environmental impact of negligent disposal of batteries, including the substantial risk of impact ignition for batteries and the long-term pollution from the battery \cite{wang2014economic}.
Moreover, most batteries are designed to be re-usable, meaning that discarding them without recovery would be wasteful.

For these reasons, application of staged batteries will likely be confined to specialized environments and circumstances.
Examples may include long-duration surveillance of, e.g., industrial sites, where discarded batteries can be periodically collected for re-use and there is no danger of injuring third parties;
use in emergencies such as in the event of a natural disaster, where the immediate value of longer-duration flight (such as transmitting warnings) out-weighs the longer-term impacts;
or in extreme environments such as for space exploration missions on other celestial bodies.

\vspace{0.1mm}
\section{Conclusion and Future Work} \label{sec:conclusion}

In this paper, we have introduced the concept of staging energy sources for UAVs, specifically considering multirotors for our analysis.
The idea is to discard the energy sources that cannot supply power anymore.
This reduces the mass of the vehicle, thereby reducing power consumption.
This, in turn, increases the overall flight time of the vehicle.

We presented a model to predict the flight time of a multi-stage multirotor with given physical parameters related to the energy source and power consumption.
We then analyzed two specific cases of optimal staging.
First, given $N$ energy storage stages of fixed, known masses, we proved that flight time can be maximized by staging them in order of decreasing mass, with the heaviest stage being depleted and ejected first.
Second, given an energy storage mass budget, finding the optimal way of partitioning it into $N$ stages to maximize the flight time.
This problem does not have a closed-form solution.
Numerical results were presented and compared with the equal staging case.

We also presented an analysis for continuous staging for two cases: internal combustion engines driving propellers, and reaction engines.
We observed that there is a fundamental limit on flight time when using I.C. engines, even when using an unlimited amount of fuel, which does not exist for reaction engines.
Their dependence on air density and gravity is also fundamentally different, which would typically lead to one of them being preferable depending on the environment.

The staging analysis and optimal order of staging were validated experimentally by conducting flight experiments on a custom-designed two-stage quadcopter with a staging mechanism to detach a used battery in-flight.
Experimental flight times matched well (within 5\%) with the predicted flight times.
The two-stage case showed a 19\% higher flight time for the same battery mass than the single-stage case (or 11\% after correcting for the additional mass of the staging mechanism).
We then discussed the environmental impact of discarding used stages, and specialized areas where the staging concept can be utilized.

Future work may include the use of additional stages to further increase the flight time. %
This would entail a more intricate design for ejection of each stage individually and for seamless switching between multiple energy sources.
Another extension could be exploring the design and control of a multirotor using reaction engines as thrusters.
This would require careful design to have actuation in the desired degrees of freedom.
A third extension could be the planning of locations for stage ejections in an industrial setting, so that collection of these ejected stages can be localized and potentially automated.

\section*{Acknowledgment}
We gratefully acknowledge financial support from NAVER LABS and Code42 Air.

The experimental testbed at the HiPeRLab is the result of contributions of many people, a full list of which can be found at \url{hiperlab.berkeley.edu/members/}.

\balance
\bibliographystyle{IEEEtran}
\bibliography{main}

\end{document}